\documentclass[preprint,12pt,authoryear]{elsarticle}

\usepackage{amssymb}
\usepackage{amsmath}
\usepackage{mathrsfs}
\usepackage{amsthm}

\newtheorem{thm}{Theorem} 
\newtheorem{lem}[thm]{Lemma} 
\newdefinition{rmk}{Remark} 
\newproof{pf}{Proof} 

\begin{document}

\begin{frontmatter}



\title{Non-existence of a  certain kind of  finite-letter mutual information characterization for a class of time-invariant Markoff channels}


\author{Mukul Agarwal 
}


\begin{abstract}
We provide a rigorous definition of a certain kind of characterization for capacity regions of a family of Markoff networks which is based on optimization problems resulting out of calculating conditional mutual informations from a finite number of random variables and by constraining these random variables in a certain way. 
This definition is partly motivated by the definition of single-letter characterizations in information theory. For a point-to-point Markoff channel, we prove that approximating the solution to these characterizations within an additive constant is a computable problem. Based on previous undecidability results concerning capacities of certain class of finite state machine channels (FSMCs), it will follow  that there exists an example of family of FSMCs for which given such a characterization, this characterization  cannot represent the capacity of this family of FSMCs.

\end{abstract}


\end{frontmatter}

\section{Introduction}
In this paper, wse provide a rigorous definition of a certain kind of characterization for capacity regions of a family of finite input, finite output, finite state Markoff networks which is based on optimization problems resulting out of calculating conditional mutual informations from a finite number of random variables taking values in finite sets, and by constraining the probability distributions corresponding to these random variables by polynomial constraints. Such a characterization, the rigorous definition of which is the subject of Section \ref{FinLC}, will hence forth be called an FLC. The definition of  FLC is partly motivated by single-letter characterizations in information theory. An admissible FLC corressponding to a class of Markoff networks  is an FLC for which, given a network in this class of networks, after substituting for variables of the FLC, the values which determine the conditional probabilities which define the network, we get an optimization problem, the closure of the feasible region of which is the capacity region of the network, and this is the case for all networks in this class of networks. Note that the FLC is independent of the particular network; however, the    variables in the FLC take values dependent on the conditional probability distributions which define the particular network.

 Let $F$ be an FLC which is admissible for a certain class $\mathbb D$ of Markoff channels. We will prove that this will imply that there exists an algorithm which takes as input $\epsilon > 0$ and the conditional distributions which define the particular network $d \in \mathbb D$, and provides as output, a real number $g \in [C_d - \epsilon, C_d + \epsilon]$.
 In other words, if there exists an admissible FLC for $\mathbb D$, approximating the capacity of a channels $\in \mathbb D$ within an additive constant is a computable problem. This is the subject of Section \ref{FLC}.
In \citep{EG}, an example of a family of FSMCs, denoted by $\mathscr{S}_{\lambda}$  is provided for which, capacity is either $\leq \frac{\lambda}{2}$ or $\geq \lambda$ for every channel in $\mathscr{S}_{\lambda}$, and  deciding, which is the case, is an undecidable problem, that is, there is no algorithm which decides whether the capacity is $\leq \frac{\lambda}{2}$ or $\geq \lambda$ for all channels $\in \mathscr{S}_{\lambda}$.  It will follow from this result and the computability result for approximating capacities if there exists an admissible FLC, stated above, that there exists no admissible FLC for  the family of  channels $\mathscr S_{\lambda}$. This is the subject of Section \ref{NE}.  

Further, based on an earlier result concerning undecidability of approximations for the emptiness problem in the theory of probabilistic finite automatas (PFAs) \citep{MHC}, we will provide a short outline of a  proof of the non-computability of the problem of approximating capacity within an additive constant for Markoff channels with partial state information. This is the subject of  \ref{A}. This proof is complete but for one missing step, see  \ref{Discc}.

\section{FLCs} \label{FinLC}
In this section, we define an FLC corresponding to a family of  Markoff  networks $\mathbb D$, and an admissible FLC corresponding to $\mathbb D$.
Consider a  network of $n$ users. The input space at User $i$ is $\mathbb A_i$ and the output space at User $i$ is $\mathbb B_i$. $\mathbb A_i, \mathbb B_i$ are assumed to be finite sets. The network also has a state which is an element of the set $\mathbb C$, a finite set. The action of the network is given by a a conditional probability $c(b_1, b_2, \ldots, b_n, s|a_1, a_2, \ldots, a_n,s')$.  This is the probability that the channel output is $b_1, b_2, \ldots, b_n$ and the state is $s$ given that the channel input is $a_1, a_2, \ldots a_n$ and the state at the previous time was $s'$. The notion of reliable communication over such a network is the subject of information theory, see for example, \cite{KG}, for discussion of capacity region of discrete memoryless networks. A achievable rate vector would be a sequence $(L_{ij}, 1 \leq i,j \leq n, i \neq j)$ where $L_{ij}$ denotes the rate of communication from User i to User j. The closure of the set of achievable rate vectors, when considered as a subjset of $\mathbb R^{n^2-n}$ is the capacity region of this network.  In this paper, we will be restricting attention to point-to-point Markoff channels for the main results. For a discussion of capacity of such channels, the reader is referred to \cite{Gallager} and \cite{EG}.

Consider a set $\mathbb D$ of such networks such that all networks in this set have the same input, output and state spaces.

In what follows, we would need the notion of computable real numbers and algebraic real numbers. A computable number is a real number that can be computed to within any desired precision by a finite, terminating algorithm, see, for example, Chapter 9 of  \citep{Minsky}. An algebraic real number is a real number which is the root of a non-zero polynomial with integer (or rational) coefficients. An algebraic number can be specified by specifying the polynomial of which it is a root and an interval $[a,b]$ to which it belongs, where $a,b$ are rational numbers.


An FLC corresponding to $\mathbb D$ consists of two elements: Representation and Constraints. These are defined as follows:
\begin{enumerate}
\item Representation: The representation has finite sets $\mathbb X_1, \mathbb X_2, \ldots, \mathbb X_k$ (for some $k$).
 Let $(X_1, X_2, \ldots, X_k)$ be a random vector on $\prod_{i=1}^k \mathbb X_i$. The representation consists of a set of equations:
\begin{align} \label{R}
\sum_{i,j=1}^n \beta_{ij}^{(r)}R_{ij} + 
\sum_{v=1}^{j_r} \alpha^{(r)}_{v} I(\vec{U}_v^{(r)}, \vec{Y}_v^{(r)}|\vec{Z}_v^{(r)})  \leq 0, 1 \leq r \leq N
\end{align} 
The above is a set of $N$ equations for some finite $N$. The superscript $r$ represents the number of the equation,  $\beta_{ij}^{(r)}, \alpha_u^{(r)}$ are computable real numbers for all $i,j,r$, and $j_r$ are integers.
  $\vec{U}_v^{(r)}$, $\vec{Y}_v^{(r)}$ and $\vec{Y}_v^{(r)}$ are all vectors with components belonging to the set $\{X_1, X_2, \ldots, X_k\}$ and such that these vectors have different $X_i$ as components.
The inequality ($\leq$) in some or all of the equations in (\ref{R}) may also be $<, >, \geq, =$. The representation is independent of the particular  $d \in \mathbb D$.
\item
Constraints: The random variables $X_1, X_2, \ldots X_n$ lead to a probability distribution on $\prod_{i=1}^k \mathbb X_i$. $p(x_1, x_2, \ldots, x_k)$ is the probability that $X_1 = x_1, X_2 = x_2, \ldots, X_k = x_k$. These are $\prod_{i=1}^k |\mathbb X_i|$ such probabilities (which satisfy constraints that probabilities add to one). Arrange these probabilties in a certain order (the order does not matter), and we get a vector of length $\prod_{i=1}^k |\mathbb X_i|$. Denote this vector by $\vec{p}$. 
The action of the network $d$ is denoted by $c_d$. Note that $c_d$ is a transition probability and corresponding to 
$b_i \in \mathbb B_i$, 
$s \in \mathbb C$, 
$a_i \in \mathbb A_i$, 
$s' \in \mathbb C$, 
we have the transition probability 
$c_d(b_1, \ldots, b_n, s|a_1, a_2, \ldots, a_n, s')$.
 These are, then, $\prod_{i=1}^n |\mathbb A_i| \times \prod_{i=1}^n | \mathbb B_i| \times |\mathbb C|^2$ values. Arrange them in some order and this leads to a vector $\vec{q}_d$ where $d$ refers to the particular network $d \in \mathbb D$. Note that the vector $\vec{p}$ is variable whereas the vector $\vec{q}_d$ is constant. The constraints are of the form
\begin{align}
\vec{f}(\vec{p}, \vec{q}_d) = 0
\end{align}
Here,  $\vec{f} = (f_1, \ldots, f_l)$  is a vector function for some finite $l$. Each $f_i$ is a polynomial in the components of $\vec{p}$ and $\vec{q_d}$. Also, the coefficients of these polynomials should be algebraic numbers. $\vec{f}$ is independent of the particular $d \in \mathbb D$.
Note again, that the only unknowns are components of $\vec{p}$, though the equations are polynomial in $\vec{p}$ and $\vec{q}_d$. This is because for the particular $d \in \mathbb D$, $q_d$ will be substituted for, by the conditional probability defining the network.
Also note that $q_d$ depends on $d \in \mathbb D$.
 The idea is that the function of $\vec{p}$ used  when determining the constraints should be independent of $\vec{q}_d$, but  different $\vec{q}_d$s in this function will result in different polynomials in $\vec{p}$.  These are the constraints on $\vec{p}$ and thus, the constraints on $(X_1, X_2, \ldots, X_k)$. Note that there are other constraints on $\vec{p}$, that the components of $\vec{p}$ are non-negative and add to $1$.
\end{enumerate}
The above pair of representation and constraints is called an FLC corresponding to $\mathbb D$. Consider the $(R_{ij}, 1 \leq i,j \leq n)$ which are feasible for the  optimization problem determined by the above Representation and Constraints, and denote by $\mathbb E_d$, the closure of this feasible  region for the network $d \in \mathbb D$. Note that this feasible region depends on $d \in \mathbb D$ because the optimization problem depends on $\vec{q}_d$, even though the FLC is independent of $\mathbb D$.
 If $\mathbb E_d$ is the same as the capacity region of $d$  $\forall d \in \mathbb D$, we call the above FLC an admissible FLC corresponding to $\mathbb D$.

See Subsection \ref{RRD} for a discussion on the motivation of this definition.


\section{Computability of approximating the capacities for a family of Markoff channels if there exists an admissible FLC} \label{FLC}
Note that for a point-to-point Markoff channel, $(R_{ij}, 1 \leq i,j \leq n)$ is just a $R$, the rate of communication from User 1 to User 2. Such an FSMC is determined by a conditional probability $c(y,s|x,s')$.
\begin{lem} \label{L2}
Consider a set $\mathbb D$ of FSMCs, all with the same input, output and state spaces, such that $c_d(y,s|x,s')$ is algebraic for all $x,y,s,s'$ and for all $d \in \mathbb D$. Let there exist an admissible FLC, denoted by $F$, corresponding to $\mathbb D$. Then, given any $\epsilon > 0$ small enough, given $d \in \mathbb D$, there exists an algorithm (independent of $d$) which outputs $\beta_d$ such that $|C_d -\beta_d| < \epsilon$ where $C_d$ is the capacity of $d$.
\end{lem}
\begin{pf}
Fix $d \in \mathbb D$. First consider the constraints of $F$. These are polynomial inequalities in $\vec{p}$.
 Denote by $\mathbb V$, the set of all these constraints. Given an algebraic number $\delta > 0$.
Construct a finite set $\mathbb U$ consisting of probability distributions which are elements of $\mathcal P(\mathbb X_1 \times \mathbb X_2 \times \cdots \times \mathbb X_n)$ such that 
for any $s \in \mathcal P(\mathbb X_1 \times \mathbb X_2 \times \cdots \times \mathbb X_n)$, $\exists u \in \mathbb U$ such that $||s-u||_1 < \delta$
and that,
$u(x_1, x_2, \ldots, x_n)$ is rational $\forall x_i \in \mathbb X_i$, $1 \leq i \leq n$, $\forall u \in \mathbb  U$.
Pick a particular $u \in \mathbb U$. Recall that $\vec{p}$ is a vector consisting of $p(x_1, x_2, \ldots, x_n)$ as $x_1, x_2, \ldots, x_n$ vary. Consider the set of constraints which includes the constraints $\mathbb V$ along with the constraints $-\delta \leq p(x_1, x_2, \ldots x_n) - u(x_1, x_2, \ldots, x_n) \leq \delta$ (we could choose, in this inequality, a number much smaller than $\delta$, but $\delta$ suffices), $\forall x_i \in \mathbb X_i, 1 \leq i \leq n$. Denote the set of all these inequalities by $\mathbb W$. The first order theory of the real closed field is decidable \citep{T}(see also \citep{T2}, \citep{T3}),  and thus, determining whether a feasible solution exists to the set of inequalities $\mathbb W$ is decidable.
Construct a set $\mathbb S$ as follows: $u \in \mathbb S$ if and only $u \in \mathbb U$ and  if there exists a feasible solution to the set of inequalities $\mathbb W$. The set $\mathbb S$ should be thought of as an `approximation' (consisting of a finite number of elements) of the feasible region corresponding to the constraints in $F$.

Next, consider the representation. The representation of an  FLC corresponding to a point-to-point Markoff channel, by noting (\ref{R}) consists of inequalities of the following form: 
\begin{align} \label{R2}
R \leq
\sum_{v=1}^{j_r} \alpha^{(r)}_{v} I(\vec{U}_v^{(r)}, \vec{Y}_v^{(r)}|\vec{Z}_v^{(r)}), 1 \leq r \leq N,
\end{align} 
in other words,
\begin{align}
R \leq \min_{r \in  \{1, 2, \ldots, N\}} \sum_{v=1}^{j_r} \alpha^{(r)}_{v} I(\vec{U}_v^{(r)}, \vec{Y}_v^{(r)}|\vec{Z}_v^{(r)})
\end{align}
The capacity of the channel is the supremum of such $R$ over the $(X_1, X_2, \ldots, X_r)$ which satisfy the constraints of $F$. Recall that $ \vec{U}_v^{(r)}, \vec{Y}_v^{(r)},\vec{Z}_v^{(r)}$ are random vectors whose components belong to the set $\{X_1, X_2, \ldots, X_r\}$.
Note that conditional mutual information is uniformly continuous and that, a quantification of this fact follows by noting  that $I(X;Y|Z)$ = $H(X,Z)+H(Y,Z)-H(Z)-H(X,Y,Z)$ and noting (6) in \citep{PS}, which is the following:
For a finite set $\mathbb G$, for $p_1, p_2 \in \mathcal P(\mathbb G)$, $||p_1-p_2||_1  \leq \frac{1}{2}$, 
\begin{align}\label{UC}
|H(p_1) - H(p_2)| \leq ||p_1-p_2||_1 \ln \frac{|\mathbb G|}{||p_1-p_2||_1}
\end{align}
It follows, since $F$ is admissible for $\mathbb D$,  by how the set $\mathbb S$ is constructed from the constraints of $F$,  by noting (\ref{UC}), that for any $d \in \mathbb D$, that there exists a $\delta > 0$ such that  
\begin{align}
C_d - \frac{\epsilon}{10} \leq \ \triangleq \max_{u \in \mathbb \mathbb S} \min_{1 \leq r \leq N} \sum_{v=1}^{j_r} \alpha^{(r)}_{v} I(\vec{U}_u^{(v)}, \vec{Y}_u^{(r)}|\vec{Z}_v^{(r)})\leq C_d + \frac{\epsilon}{10} 
\end{align}
where in the above equation, the various $I(\vec{U}_v^{(r)}, \vec{Y}_v^{(r)}|\vec{Z}_v^{(r)})$ are calculated corresponding to the particular $(X_1, X_2, \ldots, X_k)$ which is the random vector corresponding to $u \in \mathbb S$.Denote
\begin{align}
\gamma \triangleq \max_{u \in \mathbb \mathbb S} \min_{1 \leq r \leq N} \sum_{v=1}^{j_r} \alpha^{(r)}_{v} I(\vec{U}_v^{(r)}, \vec{Y}_v^{(r)}|\vec{Z}_v^{(r)})
\end{align}
There would be an error in calculation of $\gamma$ by the algorithm because of error in computation of the conditional mutual informations and error in computation of $\alpha^{(r)}_v$. The latter are computable by the definition of an FLC, and thus, can be approximated within an arbitrary accuracy by means of an algorithm. The conditional mutual informations can also be approximated within an arbitrary accuracy by means of an algorithm because all probabilities entering the calculations of these mutual informations are rational numbers; this follows because we compute mutual informations corresponding to $u \in \mathbb S$, and by construction, 
the probability distribution corresponding to $u$ has $p(x_1, x_2, \ldots, x_k)$ rational for all $x_i \in \mathbb X_i$, $1 \leq i \leq k$. 
It follows, then, that the error in computation of $\gamma$  can be made $\leq \frac{\epsilon}{10}$.
 Denote the value of $\gamma$ after this calculation error by $\beta$. This $\beta$ satisfies the properties of the $\beta$ required in the lemma.
\end{pf}

\section{Non-existence of an admissible FLC for a certain family of FSMCs } \label{NE}
\begin{thm}\label{TT}
There exists a set $\mathbb D$ of FSMCs with the same input, output and state spaces for which given any FLC, this FLC is not admissible for $\mathbb D$.
\end{thm}
\begin{pf}
Consider the set of channels 
$\mathscr S_{\lambda}
$, where the latter is defined  in \citep{EG}. Note that $c(y,s|x,s')$ (the notation in \citep{EG} is $p(y,s|x,s')$ instead of $c(y,s|x,s')$) is  rational 
$\forall$ channels $\in
\mathscr S_{\lambda}
$, $\forall, x,y,s,s'$, by construction in \citep{EG}, and thus algebraic. 
Note, further, that FSMCs in $\mathscr S_{\lambda}$ have the same input, output and state spaces (number of input symbols = $10$, output symbols = $2$ and state symbols = $62$).  Main Result $2$ in \citep{EG} states   that for a given $\lambda$, all channels in $\mathscr S_{\lambda}$ have capacity $\geq \lambda$ or $\leq \frac{\lambda}{2}$. Let $\lambda > 0$. Let there be an admissible FLC for $\mathscr S_{\lambda}$. Choose $\epsilon = \frac{\lambda}{20}$ ( $20$ in the denominator is chosen arbitrarily in a way that it is significantly less than $\frac{1}{2}$). Consider a channel $d \in \mathscr{S}_{\lambda}$ with capacity $C$.  By Lemma \ref{L2}, there exists an algorithm which says that the capacity of $d$ $\in [C-\epsilon, C+\epsilon]$. If capacity of $d$ is  $\leq \frac{\lambda}{2}$, it follows, as a consequence of the output of the algorithm,  that the capacity of d $\leq \frac{\lambda}{2} + \epsilon$ which is $<\lambda$, and thus, the capacity of this channel is $\leq \frac{\lambda}{2}$. Similarly, if capacity of $d$ is $\geq \lambda$, it follows as a consequence of the output of the algorithm, that the capacity of $d$ $\geq \lambda - \epsilon $ which is $> \frac{\lambda}{2}$, and thus, the capacity of $d$ is $\geq \lambda$.
%
Thus, if there exists an admissible FLC for the set of channels $\mathscr S_{\lambda}$, it is decidable whether the capacity of the channel is $\leq \frac{\lambda}{2}$ or the capacity of the channel $\geq {\lambda}$, and this contradicts Main Result $2$ in \citep{EG}.  The
 only possibility, thus, is that there is no admissible FLC for the class of channels $\mathscr{S}_ {\lambda}$ for $\lambda > 0$.
\end{pf}



\section{Recapitulation, discussions and research directions}

\subsection{Recapitulation}

In this paper, it was proved that finite letter characterizations as defined in this paper do not exist for a certain class of Markoff networks, that is, finite state machine channels. The idea of the proof was to use non-computability for approximating capacity for this class of channels \cite{EG}, and in this paper, we prove capacity can be approximated for FLCs by means of an algorithm. It follows that FLCs do not exist for this class of channels. In the appendix, a short proof non-computability of approximating channel capacity for Markoff channels with partial state information is given.

\subsection{Discussions and research directions} \label{RRD}

\begin{itemize}

\item
Motivation for the definition of an FLC and an admissible FLC:

The  definition of FLCs  is partly motivated by existing single letter characterizations for capacity regions of networks in information theory. Note, for example, that the capacity region of a point-to-point memoryless channel, a memoryless multiple-access channel, the Marton region for a broadcast channel \citep{KG}, can all be put in the form of the above mentioned optimization problem. We have abstracted out the properties of the definitions of these characterizations, and made them more general (in the constraints) when making this definition. (The other motivation is that with this definition, we are able to prove the theorems we proved).

As an example, consider the Han-Kobayashi region for the capacity region of the interference channel $c(y_1,y_2|x_1,x_2)$ (the notation used is from \cite{KG}):

$R_1 < I(X_1;Y_1|U_2, Q)$

$R_2 < I(X_2;Y_2|U_1, Q)$

$R_1+R_2 < I(X_1,U_2;Y_1|Q)+I(X_2;Y_2|U_1,U_2,Q)$

$R_1+R_2 <I(X_2,U_1;Y_2|Q)+I(X_1;Y_1|U_1,U_2,Q)$

$R_1+R_2 <I(X_1,U_2;Y_1|U_1,Q)+I(X_2,U_1;Y_2|U_2,Q)$

$2R_1+R_2<I(X_1,U_2;Y_1|Q)+I(X_1;Y_1|U_1,U_2,Q)+I(X_2,U_1;Y_2|U_2,Q)$

$R_1+2R_2<I(X_2,U_1;Y_2|Q)+I(X_2;Y_2|U_1,U_2,Q)+I(X_1,U_2;Y_1|U_1,Q)$,

for some PMF $p(q)p(u_1,x_1|q)p(u_2,x_2|q)$ (and some cardinality bounds on sets).

The seven random variables $Q,U_1,U_2,X_1,X_2,Y_1,Y_2$: can be renamed $X_1, X_2, \ldots, X_7$. That $p(q,u_1,x_1,u_2,x_2)$ is constrained to be of the form $p(q)p(u_1,x_1|q)p(u_2,x_2|q)$ can be written
as polynomial equations in $\vec{p}$. That $p(y_1,y_2|x_1,x_2)$ should be precisely the action of the interference channel $c(y_1,y_2|x_1,x_2)$ can
be written as polynomial constraints in $\vec{p}$ and $\vec{q}_c$. That probabilities add to one lead to polynomial equations.
Thus, achievable rates corresponding to the Han-Kobayashi region can be written in the form of Representation and 
Constraints. 

As a simpler example, consider the achievable rate  region for a DMC, given by a conditional probability $p(y|x)$ where $x \in \mathbb X$ and $y \in \mathbb Y$. This region given by $R < I(X;Y)$ where $p_X$ is a probability distribution on $\mathbb X$ and $p_{Y|X}$ is restricted to be the action of the channel. Clearly, this region can also be written in the above form of Representation and Constraints.

Next, consider the Marton's inner bound for a broadcast channel $p(y_1, y_2|x)$ given by (the notation used is from \cite{KG})

$R_1 \leq H(Y_1)$

$R_2 \leq I(U;Y_2)$

$R_1 + R_2 \leq H(Y_1|U) + I(U;Y_2)$

for some pmf $p(u|x)$.

By a reasoning similar to that for the Han-Kobayashi region described above, this region can also be in the form of Representation and Constraints. 

The same is the case for other existing characterizations in information theory.

\item
Note that in the set of channels $\mathscr S_{\lambda}$, the gap between $\frac{\lambda}{2}$ and ${\lambda}$ is a well
quantified gap (equal to $\frac{\lambda}{2}$). For this reason, we can prove, by slight modifications to the proofs of Lemma \ref{L2} 
and  Theorem \ref{TT} that there exists an $\epsilon$ small enough that no FLC even approximates the capacity of the family of
 channels $\mathscr S_{\lambda}$. In other words, there exists no FLC such that $\forall d \in \mathscr S_{\lambda}$  the feasible region of the FLC corresponding to $  d $ is  $[0, C_d + \gamma]$ or $[0, C_d - \gamma]$ for any $\gamma \leq \epsilon$ where $C_d$ is the
 capacity of $d$.

\item
 It would also be worthwhile to explore which sets of channels and networks can this result be generalized to. For example, the knowledge of initial state is necessary at the transmitter, both in \citep{EG} and the example we present in  \ref{A} for construction of Markoff channels for which an undecidability result holds in order to approximate capacity, which leads to a result for non-existence of FLCs for the class of channels $\mathcal S_{\mathcal \lambda}$ defined in \citep{EG}. However, it is unclear what happens with channels where initial state is not known at the transmitter. For example, for indecomposable channels \citep{Gallager}, channels for which the knowledge of initial state dies down with time, and thus, the knowledge of initial state is not necessary in the sense that the capacity of the channel is independent of whether the initial state is known or not, and is independent of the initial state (in the language of \citep{Gallager}, $\Bar{C} = \underbar{C}$), it is unclear whether such a result will hold, and it would be important to see, what happens for this class of channels. The authors doubt that an undecidability result for approximating capacity can be proved for indecomposable channels by using the technique in \citep{EG} or the technique in \ref{A} in this paper; that the initial state is known at the transmitter and that, the channel action might vary depending on the initial state (in the language of \citep{Gallager} $\Bar{C}$ and $\underbar{C}$ might be different) seems to be crucial in both \citep{EG} and \ref{A}. That said, the authors would also like to speculate that admissible FLCs do not exist in general, for indecomposable channels. If this is indeed the case, it needs a proof via one technique or the other, and if this is not the case, it needs to be proved that FLCs indeed exist. It the opposite is true, that is, admissible FLCs do exist for indecomposable channels, that would require a proof. In general, it would be worthwhile to explore, for which sets of Markoff channels, with or without feedback, and with or without partial information, do FLCs exist. The authors emphasize that much of this paragraph is speculation.

Further, 
 it would be important to explore whether we can prove non-computability of approximating capacity results for general networks, in particular, general memoryless networks.  A memoryless network consists of $n$ users and the action of the network can be described by a transition probability $p(y_1, y_2, \ldots, y_n|x_1, x_2, \ldots, x_n)$, where this transition probability represents the probability that the output at User $i$ is $y_i$, $1 \leq i \leq n$ if the input at User i is $x_i$, $1 \leq i \leq n$. If we can prove a non-computability result of approximation of capacity region for memoryless networks for fixed input and output spaces, that may  also imply the non-existence of FLCs for these networks (needs proof) , though the author doubts that proving such a non-computability result would be possible by the method used in \citep{EG} and and the construction in \ref{A} in this paper: that there is memory in the channel seems to be fundamental in both these cases. That said, as for the case of indecomposable channels, the authors would like to speculate that admissible FLCs do not exist, in general, for memoryless networks. If this is indeed the case, it needs proof via one technique or the other, and if this is not the case, it needs to be proved that FLCs indeed exist. Memoryless networks are the simplest kinds of networks, and a positive or negative result for existence or non-existence of FLCs for these networks will shed light on the case for general networks. It neesds to be said that whether  FLCs exist or not is not even known for the simple case of the broadcast channel The authors emphasize, as in the previous paragraph,  that much of this paragraph is speculation.

\item
From a mathematical perspective, it would be worthwhile to explore  whether the theorem in this paper can be proved without resorting to the `heavy-duty machinery' of the decidability of the first order theory of the real closed field. Also from a mathematical perspective, based on the proof of Lemma \ref{L2}, the authors conjecture that  the Representation part of the FLC, as defined in Section \ref{FinLC},  can consist of functions much more general than mutual information and further, it may also be non-linear in $(R_{ij}, 1 \leq i,j \leq N)$, and still, Lemma \ref{L2} and Theorem \ref{TT} will hold. It may be worthwhile to see, to what extent, these generalizations are possible.
\end{itemize}

\section{Acknowledgements}
The idea of using undecidability results  as a way of proving the non-existence of admissible finite letter characterizations is due to Prof. John Tsitsiklis. Further, the idea of using the fact that the first order theory of the real closed field is decidable in order to prove Lemma \ref{L2} is also due to Prof. Tsitsiklis. Beyond that, the author profusely thanks Prof. Tsitsiklis for various insightful discussions on this problem.





  \bibliographystyle{elsarticle-num} 


\appendix

\section{Non-computability of approximating capacity for time-invariant markoff channels with partial state information} \label{A}

In this appendix, we provide an outline of a proof of the non-computability of approximating, within an additive constant, the capacity of Markoff channels with partial state information by using prior results in the theory of PFAs concerning undecidability of approximations related to the emptiness problem, as stated and proved in \citep{MHC}. The idea of PFAs is also used in \citep{EG}; here, we provide a short construction and proof based on the above mentioned result in \citep{MHC}.  The authors found this construction without knowledge of the work of Elkouss et al.. The outline is a complete proof but for one missing step discussed in  \ref{Discc}. This appendix is written in discursive style.

Consider the definition of PFA in Section 2.3 in \citep{MHC}. The details of this definition are the following (we cut and paste from \citep{MHC}): A PFA, $\mathcal M$ is defined by a quintuple $\mathcal M = (Q, \Sigma, T, s_1, s_n)$ where $Q$ is a set of $n$ states, $\Sigma$ is the input alphabet, $T$ is a set of $n \times n$ row-stochastic transition matrices, one for each symbol in $\Sigma$, $s_1 \in \mathcal Q$ is the initial state of the PFA, and $s_n \in \mathcal Q$ is an accepting state. The state transition is determined as follows:
\begin{itemize}
\item
the current input symbol $a$ determines a transition matrix $\mathbb M_a$,
\item
the current state $s_i$ determines the row $\mathbb M_a[i,\cdot]$, a probability distribution over the possible next states, and
\item
the state changes according to the probability distribution $\mathbb M_a[i,\cdot]$
\end{itemize} 
The accepting state $s_n$ is absorbing, that is, $\mathbb M_a[n,n] = 1$, $\forall a \in  \Sigma$. 

Corresponding to this  PFA, we  construct an FSMC as follows: 
The channel takes as input, an element of the set $\Sigma$. The state space of the channel is $Q$, and the initial state is $s_1$. $s_n$ is the only `good' state in the sense that will become clear below. The channel acts as follows: if the channel is in state $s_i$, and the input to the channel is $a$, the channel transitions to state $s_{i'}$ with probability $\mathbb M_a[i,i']$. Partial state information of whether the channel is in state $s_n$ or not is known at both the transmitter and the receiver. Further, when the channel is in state $s_n$, transmission of $K$ bits (think of $K$ as being large) per channel use is possible over the channel, otherwise, the channel outputs an error symbol $`e'$. 


The definition of channel capacity that we use is generalized capacity, which, in this scenario, is described best by an example as stated in the second column of Page 1 of \citep{SS}. In this example, the channel is binary input, binary output, and is such that:
\begin{itemize}
\item
with probability $1-q$, the channel reproduces the input sequence error free
\item
with probability $q$, the channel introduces errors with probability $h^{-1}(0.5)$, where $h(x)$ is the binary entropy function in bits. 
\end{itemize}
If the encoder knew which of the two states is in effect, it could adapt the rate, and the average rate of communication would be $1 - \frac{q}{2}$. However, if the encoder did not know, which state is in effect, then, reliable communication would only be possible at rates $< \frac{1}{2}$. The reader is referred to \cite{SS} for details.


For the channel we have constructed above, partial state knowledge is available at the transmitter and the receiver  in the sense that the transmitter and the receiver know, whether the channel is in state $s_n$ or not. Further, if the channel enters state $s_n$, it stays in state $s_n$. Also, if the channel is in state $s_n$, $K$ bits can be transmitted reliably over the channel per channel use. Thus, once the channel enters state $s_n$, $K$ bits are reliably communicated per channel use over the channel, that point onwards.  Note, now, the definition of $L(\mathcal M, \tau)$ in \citep{MHC}. This is the set of all infinite-length strings such that the PFA ends in state $s_n$ with probability $> \tau$. The emptiness problem is: given a PFA $\mathcal M$, and given a threshold $\tau$, determine whether $L(\mathcal M, \tau)$ is empty or not. See \citep{MHC} for precise details of the definition of $L(\mathcal M, \tau)$ and for the statement of the emptiness problem. 
Note that
\begin{align}
& \Pr(\mbox{PFA enters state}\  s_n)  \nonumber \\
= &  \sum_{\alpha = 0}^\infty \Pr (\mbox{PFA enters state}\ s_n\ \mbox{at time} \ \alpha)
\end{align}
The above holds because probability of a countable union of disjoint events is equal to the sum of the probabilities of these events.
Now, the PFA enters $s_n$ with probability $> \tau$. This implies, from the above, that for any $\gamma > 0$, there exists a $\psi$ such that the probability that the PFA enters state $s_n$ at time $< \psi$ is $\geq \tau - \gamma$.

Consider the following communication scheme: if the channel does enter state $s_n$ at time $<\psi$, communicate bits reliably over the channel starting at this point, else, declare that no reliable communication is possible over the channel. 
Recall the example of generalized channel capacity discussed above.
For the channel that we have constructed from the PFA as discussed above, it follows, by the discussion above, that  if $L(\mathcal M, \tau_1)$ is non-empty,  reliable communication can be accomplished over this channel at a rate $R > \tau_1 K - \kappa$ for any $\kappa > 0$ by use of the above coding scheme (adapted to $\tau = \tau_1$).  During the time slots when the channel is not in state $s_n$, at most $\log |\Sigma|$ bits per channel use can be communicated over the channel by sending bits by coding them into the channel input and by noting the partial channel state at the output.
 It follows, that if  $L(\mathcal M, \tau_2)$ is empty, at most $\tau_2 K + \log|\Sigma|$ bits can be communicated reliably over the channel per channel use, irrespective of the coding scheme. Assume that $K$ is much larger compared to  $|Q|$ and $|\Sigma|$ and assume in what follows that $\delta_1, \delta_2$ are small.

We state Corollary 3.4 from \citep{MHC}; this is the key result on which our result concerning uncomputability of approximating capacity will be based:
For any fixed $\epsilon$, $0 < \epsilon < 1$, the following is undecidable: given a PFA for which one of the two cases hold: 
\begin{itemize}
\item
(1) the PFA accepts some string with probability $> 1-\epsilon$.
\item
(2) the PFA accepts no string with probability $> \epsilon$.
\end{itemize}
Deciding, whether case (1) holds, is undecidable.
\begin{lem} \label{UMDP}
Approximating the capacity of a Markoff channel with partial state information at the encoder and decoder within an additive constant is an uncomputable problem in general.
\end{lem}
\begin{proof}
Consider the channels constructed from  PFAs considered in Corollary 3.4 in \citep{MHC} by following the procedure mentioned in this appendix, with $\epsilon$ in Corollary 3.4 in \citep{MHC} made equal to $\delta_1$. Denote $C_l \triangleq \delta_1 K + \log | \Sigma|$ and denote $C_u \triangleq (1-\delta_1)K - \delta_2$. Note that $C_u$ is strictly greater than $C_l$ for sufficiently large $K$.  It follows, by the above explanation, that if the PFA in Corollary 3.4 in \citep{MHC} falls under case (1), then the capacity of the channel $> C_u$. Similarly, if the PFA in Corollary 3.4 in \citep{MHC} falls under case (2), then the capacity of the channel $<C_l$. Note here, that in this channel coding problem, the information whether the state is $s_n$ or not is available at the transmitter. This changes the underlying PFA a little in the sense that there is partial knowlege of the state, that is, whether  state is $s_n$ or not, and this can possibly increase the $\tau$ for which $L(\mathcal M, \tau)$ is non-empty, and thus, 
 it can possibly increase the capacity of the channel by making $a$ depend on whether the channel is in state $s_n$ or not. This however is not the case because consider a string which is used when there is partial knowledge of the state. This string will end when the state enters $s_n$, but for the case when the state has not entered $s_n$, the string will continue on for possibly infinite length of time. Use this same string when there is no partial knowledge of state in the PFA, and it follows that if $L(\mathcal M, \tau)$ is non-empty if there is partial state information of the above form for the PFA, then $L(\mathcal M, \tau)$ is non-empty for the original PFA too.
 With this explanation, denote $\Delta \triangleq \frac{C_u-C_l}{8}$. It then follows, by Corollary 3.4 in \citep{MHC}, that the capacity of the channel cannot be approximated with an additive constant of $\Delta$ by means of an algorithm.
This is because, if the contrary were true, we would be able to say, whether the capacity of the channel is $< C_l + 2 \Delta$ or $> C_u - 2 \Delta$. It follows, then, that we would be able to say whether the capacity of the channel is $<C_l$ or $>C_u$ (because at least one of these holds), from which, we would be able to say for the PFA, whether case (1) or case (2) holds in Corollary 3.4 in \citep{MHC}, and this is undecidable by Corollary 3.4 in \citep{MHC}.
\end{proof}
\subsection{Discussion} \label{Discc}
Note that in this construction, it is unclear if the PFA has fixed cardinality in state space and input space, and thus, it is unclear if the channel thus constructed has a fixed state space and input space. However, it is possible that results from \citep{Blondel} where undecidability results are proved for PFAs of fixed dimension can be modified in order to get results for undecidability of approximations, and if this is the case, it is possible that an analogue of Theorem \ref{TT} may also hold for channels which are constructed from PFAs as described in this appendix. These are just ideas at this point which might be correct or incorrect.

\end{document}